\documentclass[journal,twoside,web]{ieeecolor}
\usepackage{lcsys}
\usepackage{graphicx,color,amsmath,mathtools}
\usepackage[varvw,vvarbb]{newtxmath}
 \usepackage{mathrsfs}
\usepackage[tight]{subfigure}

\newcommand\rev[1]{{\color{black}#1}}

\DeclareMathOperator\sign{sgn}
\DeclareMathOperator\Indfun{\mathbb{1}}
\newcommand\amax{\ensuremath{\bar{a}}}
\newcommand\vmax{\ensuremath{\bar{v}}}
\newcommand\abs[1]{\ensuremath{\lvert#1\rvert}}

\newcommand\mc[1]{\mathcal{#1}}
\newtheorem{theorem}{Theorem}[section]
\newtheorem{lemma}[theorem]{Lemma}
\newtheorem{corollary}[theorem]{Corollary}
\newtheorem{proposition}[theorem]{Proposition}
\newtheorem{assumption}{Assumption}
\newtheorem{prob}{Problem}
\newtheorem{remark}{Remark}[section]

\newtheorem{definition}{Definition}

\makeatletter
\xdef\@endgadget#1{{\unskip\nobreak\hfil\penalty50\hskip1em\hbox{}\nobreak\hfil#1\parfillskip=0pt\finalhyphendemerits=0\par}}
\newcommand\@Endofsymbol{$\triangledown$}
\newcommand\Endofremark{\@endgadget{\@Endofsymbol}}
\makeatother
\title{Transient-Safe Platooning via Dynamic Headway}

\author{Gal Barkai, Jérémie Kreiss, Vineeth S. Varma, Irinel-Constantin Mor\u{a}rescu\thanks{Université de Lorraine, CNRS, CRAN, F-54000 Nancy, France. Emails: \textsl{\{gal.barkai, jeremie.kreiss, vineeth.satheeskumar-varma, constantin.morarescu\}@univ-lorraine.fr}.}\thanks{Work supported by ANR under grant COMMITS ANR-23-CE25-0005.} 
}

\begin{document}
\maketitle
\thispagestyle{empty}

\begin{abstract}
Managing autonomous vehicle platoons requires a delicate balance between string stability and rigorous safety. This challenge is intensified by aggressive transients, such as highway merging. Although Constant Time Headway (CTH) spacing is the industry standard for Cooperative Adaptive Cruise Control, it lacks formal safety guarantees during significant velocity deviations. This letter proposes a computationally efficient control framework that considers a linear time-invariant model for the dynamics of each vehicle, while ensuring formal transient safety and stability. By introducing a spacing policy that naturally converges to CTH at steady state, we establish platoon safety as an inductive property. We derive a non-linear and saturated control law for the lead follower and provide sufficient initial conditions to guarantee velocity non-negativity and safety throughout the platoon for any CTH-based followers' control law. Numerical examples indicate the proposed methodology may be applicable even under non-nominal setups.
\end{abstract}

\section{Introduction} \label{sec:intro}

For decades, autonomous vehicle platooning has faced a fundamental control challenge: simultaneously ensuring safe operation and coordinated multi-agent behavior. Early studies addressed this by developing headway policies to guarantee collision avoidance~\cite{SBS:79}. However, subsequent research revealed that disturbances from the leader can amplify down the line if not properly accounted for, causing severe velocity and spacing oscillations known as the "slinky effect"~\cite{CI:92}. To counter this, the concept of string stability was introduced, which remains fundamental across various multi-agent domains~\cite{SH:96}.

Modern strategies generally follow two paradigms. The first utilizes velocity-dependent spacing, notably the Constant Time Headway (CTH) policy~\cite{CI:92, KM:09}, which scales inter-vehicle distances with speed to dissipate disturbances. Implemented via Adaptive Cruise Control (ACC) or communication-backed Cooperative ACC (CACC), CTH robustly maintains string stability across varying speed profiles. Conversely, constant spacing policies maintain fixed geometries regardless of velocity but require a broader information flow, typically needing leader-to-all communication alongside predecessor data~\cite{SH:96}. Crucially, neither paradigm is safe under arbitrary maneuvers.  However, safety-guaranteeing policies based on quadratic velocity-dependent spacing rules, such as~\cite{SBS:79}, recover the CTH structure under synchronized velocities~\cite{IC:93}. These formulations are closely related to modern Responsibility-Sensitive Safety (RSS) safe-distance approaches and established CTH as the industry standard. Consequently, contemporary CACC research largely prioritizes string stability under communication constraints~\cite{DPH:17,SB-L:24, MPD:25}, often relying on empirical tuning for safety during minor velocity deviations~\cite{PSNWN:11,MSSNKN:14, MPD:24}.

However, the requirement for synchronized velocities is violated during aggressive transients, such as highway merges where a vehicle enters the flow at a significantly lower speed. Although recent literature addresses safety-critical platooning using control barrier functions (CBFs) \cite{GH:24, CTJM:24}, adaptive online verification \cite{MA:24}, or artificial potential functions \cite{LS-KPAN:18}, these methods often introduce a significant computational load due to high frequency re-computations in vehicular systems and do not ensure feasibility. For example, in \cite{mehra2015adaptive}, which implements CBF on an experimental testbed, a quadratic programming is solved at a 200Hz frequency.
Moreover, other safety-oriented works typically rely on static spacing policies such as constant distance \cite{LS-KPAN:18, CTJM:24} or CTH \cite{MSSNKN:14, GH:24}. These approaches do not account for the transient phase where large velocity mismatches occur.

The need for computationally efficient, rigorous safety is most acute in railway-based Automated Transit Networks (ATNs), such as Urbanloop (https://urbanloop.fr) \cite{SLSM:25} or WVU PRT \cite{wvu_prt_2025}. These systems utilize homogeneous pods with simple LTI dynamics but operate under strict operational constraints: vehicles cannot move backward, cruising velocities may be three times greater than the merging speed, and the merging point is fixed by the railway-lines. On the other hand, these systems also have several advantages over classical road vehicles such as: a centralized low-latency communication network between all pods and electrical systems with low actuation lags that can be approximated by double-integrator. Standard CACC controllers cannot guarantee velocity non-negativity during the aggressive transients induced by merging pods, while general safety-guaranteed controllers often overlook the system's inherent simplicity, leading to unnecessary computational overhead.

In this letter, we propose a control method that retains the simplicity of LTI-based CACC while providing formal transient safety guarantees. Motivated by railway-based ATNs like Urbanloop, where vehicles merge onto a main track with significant speed mismatches, we treat platoon safety as an \emph{inductive property}. By analyzing a vehicle merging in front of a synchronized platoon, we prove that, \rev{under nominal conditions,} safety is preserved for the entire string if the first follower (the original leader) maintains a non-negative velocity. To this end, we design a nonlinear control law for the first follower that ensures safety during the merge while naturally recovering CTH at steady state. We further derive sufficient conditions on the initial states that serve as a \emph{safety-admissibility criterion} to guarantee both collision avoidance and velocity non-negativity. Finally, the proposed framework preserves the standard CACC structure for the followers during the transient phase and asymptotically recovers a standard CACC configuration for the complete platoon, allowing the classical string-stability properties of the underlying CACC architecture to be recovered after the merge.

The remainder of this letter is organized as follows: Section~\ref{sec:model} formalizes the platoon model and control problem. Section~\ref{sec:inherent-safety} identifies safety properties via worst-case braking analysis. Section~\ref{sec:main} presents the nonlinear control design for the lead follower and proves the resulting safety and stability properties. Finally, Section~\ref{sec:ex} provides numerical validation, followed by concluding remarks in Section~\ref{sec:conc}.

\paragraph*{Notation} We denote by $\mathbb{R}$, $\mathbb{R}_{\geq 0}$, and $\mathbb{R}_{>0}$ the sets of real, non-negative real, and positive real numbers, respectively. We denote by $\mathbb{Z}$ the set of integers and by $\mathbb{N}$ the set of natural numbers. For $a, b \in \mathbb{Z}$ with $a \leq b$, $\llbracket a, b \rrbracket$ denotes the discrete interval $[a, b] \cap \mathbb{Z}$. For a scalar $x \in \mathbb{R}$, the sign function is defined as $\sign(x) = 1$ if $x > 0$, $-1$ if $x < 0$, and $0$ if $x = 0$. Given two Euclidean spaces $U$ and $V$, $\mc{PC}(U,V)$ is the set of piecewise continuous functions from $U$ to $V$. Given a signal $u$, the signal $\sigma_{\amax} (u)$ is its symmetric saturation to level $\amax$.

\section{System Modeling and Problem Formulation} \label{sec:model}

This section formalizes the platoon model and defines the safety and stability requirements. While motivated by electric vehicle like Urbanloop pods \cite{wang2026velocity}, the architecture is general enough to accommodate various vehicular systems. We first introduce the double-integrator dynamics and nominal CACC configuration, then characterize the initialization of a merging leader at fixed junctions. Finally, we establish the rigorous definition for safety and formulate the control problem.
%
%
We consider a mass transit system of small, autonomous electric pods operating on interconnected rails. Unlike traditional trains, these pods maintain low headways to form platoons, maximizing throughput and energy efficiency.

In this paper, we consider a platoon of $N+1$ identical vehicles $\Sigma_i$ for $i\in\llbracket 0,N\rrbracket$. In this configuration, $\Sigma_0$ acts as an uncoordinated leader merging ahead of the internal platoon $(\Sigma_1, \dots, \Sigma_N)$, as illustrated in Fig.~\ref{fig:problem_setup}.

\begin{figure}[h!]
    \centering
    \includegraphics[width=\columnwidth]{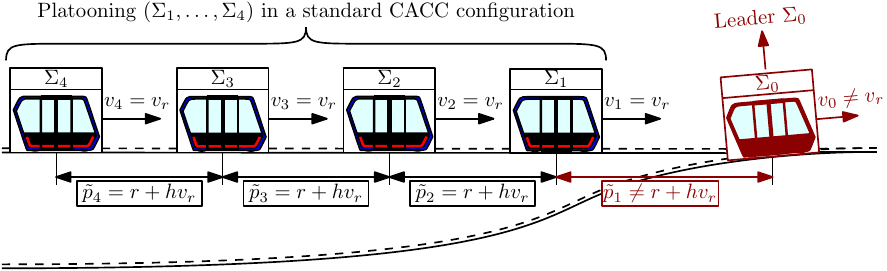}
    \caption{Urbanloop platoon initialization during a cut-in maneuver.}
    \label{fig:problem_setup}
\end{figure}

These pods utilize electric propulsion systems capable of near-instantaneous torque tracking. This high-bandwidth actuation allows the electrical dynamics to be decoupled from the mechanical kinematics \cite{wang2026velocity}. Consequently, the longitudinal dynamics of the $i$-th vehicle are accurately captured by a double-integrator model:
\begin{equation}\label{eq:vehicle}
    \Sigma_i \; :\;  
    \dot{p}_i = v_i, \quad
    \dot{v}_i = u_i
\end{equation}
where $p_i,v_i,u_i \in \mathbb{R}$ denotes the position, velocity and control input (acceleration), respectively. In nominal operation, the system is subject to physical constraints such as the acceleration is limited to $u_i \in [-\amax, \amax]$ and the velocity remains non-negative and bounded as $v_i \in [0, \vmax]$, where the constants $\vmax, \amax > 0$ represent the physical limits of the system.

Initially, the followers $(\Sigma_1,\dots,\Sigma_N)$ employ a CACC strategy to maintain CTH. Given a standstill distance $r > 0$ and time gap $h > 0$, the inter-vehicle distance $\tilde{p}_i$, relative velocity $\tilde{v}_i$, and spacing error $e_i$ for all $i \in \llbracket 1, N \rrbracket$ are defined as:
\begin{equation}
    \label{eq:spacing-error}
    \tilde p_i \coloneqq p_{i-1} - p_i, \quad  \tilde v_i \coloneqq v_{i-1} - v_i, \quad e_i \coloneqq \tilde p_i - r - h v_i.
\end{equation}
The CACC objective is to drive $e_i$ to zero while achieving velocity synchronization across the string. To formalize the state of the internal platoon prior to the disturbance introduced by the merging leader $\Sigma_0$, we define its nominal equilibrium.

\begin{definition}\label{def:Standard CACC Configuration}
   A platoon ($\Sigma_k, \dots, \Sigma_l$) with $k < l-1$ is said to be in a \emph{standard CACC configuration} when:
    \begin{equation}
        \label{eq:CACCeq}
            \tilde v_i = 0,\quad
            e_i = 0, \quad \forall i \in \llbracket k+1, l \rrbracket.
    \end{equation}
\end{definition}

Let $v_i(t)$ and $e_i(t)$ denote the time-domain trajectories of $\Sigma_i$ for $t \geq 0$, with their respective initial values denoted by $v_{i,0}$ and $e_{i,0}$. To isolate the transient safety problem of the merging pair $(\Sigma_0, \Sigma_1)$, we assume the rest of the platoon is already operating at equilibrium for a given reference velocity $v_r \in [0, \bar{v}]$, as depicted in Fig.~\ref{fig:problem_setup}, and formalized in the following assumption.
\begin{assumption}[Nominal Initialization] \label{asm:CACC-config}
    The internal platoon $(\Sigma_1,\dots,\Sigma_N)$ is initialized in a standard CACC configuration. Consequently, the initial conditions satisfy $v_{i,0} = v_r$ for all $i \in \llbracket 1, N \rrbracket$, and $e_{i,0} = 0$ for all $i \in \llbracket 2, N \rrbracket$.
\end{assumption}

The central problem addresses a ``cut-in'' maneuver (see Fig.~\ref{fig:problem_setup}), where a new leading vehicle $\Sigma_0$ merges ahead of the platoon. Because the pods operate on fixed rails they are constrained to merge at specific physical junctions and, as discussed in the introduction, the merging pod may exhibits a substantial mismatch with the cruising velocity. Consequently, the leader-follower pair $(\Sigma_0,\Sigma_1)$ is \emph{not} in a standard CACC configuration upon merging. This corresponds to an arbitrary initial leader velocity $v_{0,0} \in [0, \vmax]$ and an uncoordinated initial positive spacing error $\tilde p_{1,0}>0$.

To ground the analysis in physical reality, we formally bound the uncoordinated leader's admissible trajectories.
\begin{assumption}
    \label{asm:leader_bounds}
    Given a set of admissible initial conditions $(\tilde p_{1,0}, v_{0,0}) \in \mathbb{R}_{>0} \times [0, \vmax]$, the input trajectory $u_0(\cdot) \in \mc{PC}(\mathbb{R}_{\geq 0}, [-\amax, \amax])$ of the leader is such that its velocity remains within the physical capabilities of the system, i.e., $v_0(t) \in [0, \vmax]$ for all $t \geq 0$.
\end{assumption}
The main motivation of this letter is to ensure safety during transient phases induced by a new vehicle merging in front, where safety is defined in the following way.

\begin{definition}[Safety]
    \label{def:safety}
    For a given safety distance $d_s\in(0,r)$, a platoon $(\Sigma_k,\dots,\Sigma_l)$ with $k<l$ is said to be \emph{safe} when, evaluated along the system trajectories, the inter-vehicle distances satisfy $\tilde p_i(t) \geq d_s$ for all $t \geq 0$ and $i \in \llbracket k+1, l\rrbracket$.
\end{definition}
\begin{prob}\label{pb}
The objective of this work is threefold:
\begin{enumerate}
    \item characterize the set of initial conditions for the merging configuration under which there exists a feasible control law $u_1(\cdot)$ ensuring safety of the pair $(\Sigma_0,\Sigma_1)$;
    \item provide suitable conditions on the control laws $u_2(\cdot),\dots,u_N(\cdot)$ of the followers under which safety of the pair $(\Sigma_0,\Sigma_1)$ implies safety of the entire platoon;
    \item design a state-feedback control law $u_1(\cdot)$ ensuring safety of the platoon $(\Sigma_0,\dots,\Sigma_N)$ while asymptotically recovering a standard CACC configuration.
\end{enumerate}
\end{prob}
%
\section{Worst-case Braking and Inductive Safety}
\label{sec:inherent-safety}

This section addresses items 1) and 2) of Problem \ref{pb}. First it characterizes admissible merging configurations through a safe-set analysis of the pair $(\Sigma_0,\Sigma_1)$ under worst-case braking of $\Sigma_0$, then derives conditions under which this safety property extends inductively to the entire platoon.

\vspace{-0.25cm}
\subsection{Worst-Case Braking and the Safe Set}
Consider a pair $(\Sigma_{i-1}, \Sigma_i)$, $i\in\llbracket 1,N\rrbracket$. To guarantee safety, the pair must always reside in a state from which a collision can be avoided even under a worst-case scenario. We define this critical scenario as the preceding vehicle $\Sigma_{i-1}$ applying maximum deceleration until standstill. The safest response for $\Sigma_i$ is to likewise apply its maximum braking authority. 

For the purpose of analysis, suppose that this coordinated emergency stop is initiated at $t=0$ from an arbitrary state $(\tilde p_{i,0}, v_{i-1,0}, v_{i,0})$. The resulting control inputs are:
\begin{equation} \label{eq:u_i}
    u_k(t) = -\amax(\Indfun_{\{t \geq 0\}}-\Indfun_{\{t \geq \tau_k\}}),\quad k\in\llbracket i-1,i\rrbracket,
\end{equation}
where $\tau_k \coloneqq v_{k,0}/\amax$ denotes the time required for the $k$-th vehicle to reach standstill, and $\Indfun_{\{t \geq 0\}}$ represents the Heaviside step function. If this maneuver causes the inter-vehicle distance to fall below the threshold $d_s$, the state at which the braking commenced was fundamentally unsafe.
\begin{lemma}  \label{lem:sign_invariance}
    Let $\tilde v_{i,0} \coloneqq v_{i-1,0} - v_{i,0}$. Under the worst-case braking inputs \eqref{eq:u_i}, the relative velocity $\tilde v_i(\cdot)$ maintains the same sign depending solely on the initial conditions:
    \begin{equation}
        \sign(\tilde v_i(t)) = \sign(\tilde v_{i,0}), \quad \forall t \in [0, \max(\tau_{i-1},\tau_i)).
    \end{equation}
    Consequently, the distance $\tilde p_i(\cdot)$ is monotonic on $\mathbb{R}_{\geq0}$.
\end{lemma}
\begin{proof}
    Let $k^\star \in \arg\min_{k \in \{i-1,i\}} \tau_k$ and $\bar k^\star = 2i-1-k^\star$ denote the indices of the vehicles reaching standstill first and second, respectively. Integrating~\eqref{eq:u_i} over the three braking phases yields the following expression for the relative velocity:
    \begin{equation}
        \label{eq:delta_v_traj}
        \tilde v_i(t) = \begin{cases} 
            \tilde v_{i,0}, & 0 \leq t < \tau_{k^\star} \\
            (-1)^{i - k^\star} \amax (t - \tau_{\bar k^\star}), & \tau_{k^\star}\leq t < \tau_{\bar k^\star} \\
            0, &  t \geq  \tau_{\bar k^\star}
        \end{cases}
    \end{equation}
    During the coordinated braking phase ($0 \leq t < \tau_{k^\star}$), $\tilde v_i(t)$ is constant. During the single-vehicle braking phase ($\tau_{k^\star} \leq t < \tau_{\bar k^\star}$), only the vehicle $\Sigma_{\bar k^\star}$ continues to decelerate at rate $\amax$. As seen in \eqref{eq:delta_v_traj}, the magnitude $|\tilde v_i(t)|$ strictly decreases toward zero without ever crossing the origin, preserving its initial sign. Since the derivative $\dot{\tilde p}_i = \tilde v_i$ does not change sign, $\tilde p_i(t)$ is monotonically increasing if $v_{i-1,0} \geq v_{i,0}$ and vice versa.
\end{proof}
To quantify the gap required to absorb the excess kinetic energy of the follower when $v_{i} > v_{i-1}$, we define the \emph{quadratic safety margin} $q_i$ and the corresponding \emph{safety function} $s_i$ as:
\begin{equation}   \label{eq:safety-function}
        q_i \coloneqq \frac{v_i^2 - v_{i-1}^2}{2\amax},  \quad 
        s_i \coloneqq \tilde p_i - d_s - \max\left(0,q_i\right),
\end{equation}
for all $i\in\llbracket 2,N\rrbracket$, 
from which $q_{i,0},s_{i,0}$ are derived using initial positions and velocities.
We can now formally define the conditions under which the pair $(\Sigma_{i-1},\Sigma_i)$ can be safe against any arbitrary leader braking maneuver.
\begin{proposition} \label{prop:safe_set}
    For an arbitrary state $(\tilde p_{i,0}, v_{i-1,0}, v_{i,0}) \in \mathbb{R}_{>0} \times [0, \vmax]^2$, the following statements are equivalent:
	{ \renewcommand{\labelenumi}{(\roman{enumi})}
    \begin{enumerate}
        \item The pair $(\Sigma_{i-1},\Sigma_i)$ is safe under the worst-case braking actions~\eqref{eq:u_i}, i.e., $\tilde p_i(t) \geq d_s$ for all $t\geq 0$;
        \item The initial state $(\tilde p_{i,0}, v_{i-1,0}, v_{i,0})$ belongs to the \emph{safe set}:
        \begin{equation}
            \label{eq:safe-set}
            \mathcal{S}_i \coloneqq \left\{ (\tilde p_i, v_{i-1}, v_i) \in \mathbb{R}_{>0} \times [0, \vmax]^2 \;\middle|\; s_i \geq 0 \right\}.
        \end{equation}
    \end{enumerate}}
\end{proposition}

\begin{proof}
    Statement (i) holds if and only if $\min_{t \geq 0} \tilde p_i(t) \geq d_s$. By Lemma~\ref{lem:sign_invariance}, $\tilde p_i(t)$ is monotonic, meaning its global minimum occurs either at $t=0$ or at the end of the braking maneuver $t = \tau_{1-k^\star}$. The final distance is the initial gap modified by the difference in kinematic braking distances $\tilde p_i(\tau_{\bar k^\star}) = \tilde p_{i,0} -q_{i,0}$. Thus, the minimum distance reads $\min_{t \geq 0} \tilde p_i(t) = \tilde p_{i,0} + \min\left(0,  -q_{i,0}\right).$ Using the identity $\min(0, -q_{i,0}) = -\max(0, q_{i,0})$, enforcing $\min_{t \geq 0} \tilde p_i(t) \geq d_s$ yields the equivalent inequality $\tilde p_{i,0} - \max\left(0,  q_{i,0}\right) \geq d_s.$ By the definition of the safety function \eqref{eq:safety-function}, this inequality is equivalent to $s_{i,0} \geq 0$, which is precisely the condition for the state to belong to the safe set $\mathcal{S}_i$.
\end{proof}
Note that the quadratic term \eqref{eq:safety-function} is used to characterize the set of initial states from which a safe trajectory exists while respecting the acceleration bounds. Since the leader's braking maneuver can begin at any point, Proposition \ref{prop:safe_set} is a non-conservative characterization of the safe set.  
\vspace{-0.4cm}

\subsection{Inductive Safety of a CACC String}
We now analyze whether a platoon in a CACC configuration is nominally safe. The following Lemma establishes that this is indeed the case provided that the leader's velocity is non-negative and that the CTH errors are autonomous (i.e. nominal conditions with exact feedforward).
\begin{lemma} \label{lem:posVelSafe}
Consider a platoon $(\Sigma_1,\dots,\Sigma_N)$ in the standard CACC configuration as in Definition~\ref{def:Standard CACC Configuration}. If, for all $i\in\llbracket 2,N\rrbracket$, the control laws $u_i$ render the dynamics of $e_i$ autonomous with an asymptotically stable equilibrium, and if $v_1(t)\geq 0$ for all $t\geq 0$, then
\[
\tilde p_i(t)\geq r>d_s,
\quad
\forall t\geq 0,\;
\forall i\in\llbracket 2,N\rrbracket.
\]
\end{lemma}
\begin{proof}
Since $r>d_s$, the condition $s_{i,0} > 0$ follows directly from the nominal CACC configuration (\ref{asm:CACC-config}). For any pair $(\Sigma_{i-1}, \Sigma_{i})$, $i\in\llbracket 2,N\rrbracket$, the autonomous error $e_i$ being asymptotically stable and initialized to zero implies $e_i(t)=0$, $\forall t\geq 0$, such that, using~\eqref{eq:spacing-error}, we have $v_{i}(t)=(\tilde{p}_{i}(t)-r)/h$ and $u_{i}(t)=\tilde v_{i}(t)/h$. While $\tilde v_i\geq 0$, $\dot s_i\geq 0$ and safety is maintained. If $\tilde v_i <0$, $v_i$ decreases. Since $r>d_s$, continuity requires $v_i$ to reach zero before $\tilde p_i$ reaches $d_s$. At $v_i = 0$, $u_i =v_{i-1}/h\geq0$ (as long as $v_{i-1}>0$), preventing $v_i$ from becoming negative and ensuring $\tilde p_i\geq r>d_s$. The result follows by induction.
\end{proof}

Lemma~\ref{lem:posVelSafe} shows that for a platoon in standard CACC configuration, safety is an inductive property. The assumptions in Lemma~\ref{lem:posVelSafe} are standard requirements of string-stable CACC laws (e.g. \cite{PvdWN:14}), with the caveat that the autonomy of $e_i$ implies perfect V2V communication to enable decoupling via acceleration feedforward. We therefore adopt the following for the remainder of this letter:

\begin{assumption}[Nominal CACC Stability] \label{asm:cacc_stability}
The control laws $u_i$ for $i \in \llbracket 2, N \rrbracket$ ensure that each CTH error $e_i$ is autonomous and asymptotically stable.
\end{assumption}
\begin{remark}[Velocity non-negativity]
 One of the key property used in proving Lemma \ref{lem:posVelSafe} is the non-negativity of velocities. When this still holds, similar arguments can still be used even if the errors are not autonomous but only input-to-state stable with sufficient decay rate. This requires additional assumptions such as strict string stability and that $r$ is chosen sufficiently large with respect to the error bound. The analysis of this scenario is out of the scope of this letter, and is left for future work. 
 Inductive safety is demonstrated numerically in Section \ref{sec:ex} for a non ideal scenario. \Endofremark
\end{remark}

\vspace{-0.3cm}
\section{Control Design for the First Follower} \label{sec:main}
As established in Section~\ref{sec:inherent-safety} via Lemma \ref{lem:posVelSafe}, the safety of the entire string can be guaranteed if the first pair $(\Sigma_0,\Sigma_1)$ remains safe and $\Sigma_1$ maintains non-negative velocity. Consequently, Problem \ref{pb} is solved by redesigning $u_1$ to ensure:
\begin{enumerate}
\item \emph{Transient Safety:} Render the safe set $\mathcal{S}_1$ defined in~\eqref{eq:safe-set} forward invariant, ensuring collision avoidance ($s_1(t) \geq 0$) and velocity non-negativity ($v_1(t) \geq 0$) against any admissible leader trajectory $u_0(\cdot)$.
\item \emph{Asymptotic Recovery:} Ensure that the leader-follower pair $(\Sigma_0,\Sigma_1)$ asymptotically recovers the standard CACC configuration (i.e., driving $\tilde{v}_1(\cdot)$ and $e_1(\cdot)$ to the origin), provided the leader settles to a constant speed.
\end{enumerate}
By fulfilling these criteria, the control law $u_1(\cdot)$ serves as a safety buffer for the augmented platoon, allowing subsequent vehicles to operate under standard CACC laws while maintaining string-wide safety.
\vspace{-0.35cm}
\subsection{Proposed Controller}
To simultaneously address both considerations, we introduce an augmented spacing error $\bar e_1$ that explicitly incorporates the quadratic safety margin $q_1$:
\begin{equation}
    \label{eq:AugErr}
    \bar e_1 \coloneqq \tilde{p}_1 - r - h v_1 - q_1.
\end{equation}
This formulation naturally bridges the nominal tracking task and the transient safety constraint through two key properties. First, when the velocities are synchronized ($\tilde v_1 = 0$), the safety margin $q_1$ vanishes and $\bar e_1$ reduces to the standard CTH spacing error $e_1$ defined in~\eqref{eq:spacing-error}. Second, the inclusion of the margin $q_1$ ensures that the required safety distance is naturally incorporated into the tracking error.

To stabilize this augmented error, we propose the following control law:
\begin{equation} \label{eq:u1_full}
u_1 = \sigma_{\amax}\left(\frac{1}{h} \left( \tilde{v}_1 + \lambda \bar{e}_1 +  \frac{h v_0 u_0 - v_1\tilde v_1}{h\amax + v_1} \right)\right),
\end{equation}
where $\lambda>0$ is a design parameter and $\sigma_{\amax}(\cdot)$ denotes the symmetric saturation with level $\amax$. Note that~\eqref{eq:u1_full} is well-defined and continuous over the domain $v_1 > -h\amax$. Moreover, while the saturation is inactive \eqref{eq:u1_full} perfectly compensates for the kinematic disturbances from the leader while enforcing dissipative error dynamics.
\vspace{-0.25cm}
\subsection{Fundamental Closed-loop Properties}
The following lemma formally proves that the control law is globally well-posed for all safe initial states and admissible trajectories of $\Sigma_0$. Moreover, it provides a sufficient condition to ensure that $v_1(t)\geq 0$ for all admissible trajectories of $\Sigma_0$.
\begin{lemma}\label{lem:well_posed}
    Consider the pair $(\Sigma_0,\Sigma_1)$ initialized such that $s_{1,0}\geq 0$ and $v_{1,0}\in[0,\vmax]$, and assume that $\Sigma_1$ is controlled by~\eqref{eq:u1_full}. Under Assumption~\ref{asm:leader_bounds}, the velocity $v_1(t)$ remains lower bounded by some constant $v^\star>-\amax h$ for all $t\geq 0$. Moreover, if $\bar e_1(0)\geq 0$, then $v_1(t)\geq 0$ for all $t\geq 0$.
\end{lemma}
\begin{proof}
   Differentiating~\eqref{eq:AugErr} along the trajectories of $(\Sigma_0,\Sigma_1)$, substituting the margin derivative $\dot{q}_1 = (v_1 u_1 - v_0 u_0)/\amax$, and factoring the terms associated with the control input $u_1$ gives:
    \begin{equation*}
        \dot{\bar{e}}_1 = \tilde{v}_1 + \frac{v_0 u_0}{\amax} - \left( \frac{\amax h + v_1}{\amax} \right) u_1.
    \end{equation*}    
    We first establish that $v_1\geq v^\star> -\amax h$.
    The proof is divided into two cases depending on the saturation regime of~\eqref{eq:u1_full}.

    First, assume that the control does not saturate. Substituting \eqref{eq:u1_full} into the previous expression cancels the kinematic feedforward terms ($\tilde{v}_1$ and $v_0 u_0/\amax$), yielding
    \begin{equation} \label{eq:e1_dot}
        \dot{\bar{e}}_1 = -\lambda\left(1+\frac{v_1}{\amax h}\right) \bar{e}_1.
    \end{equation}
    Whenever $v_1> -\amax h$, the term $(1+v_1/(\amax h))$ is strictly positive. Since $\lambda>0$, the dynamics are dissipative, and $|\bar{e}_1(t)| \leq |\bar{e}_1(0)|$. Moreover, using standard integrating factor methods \eqref{eq:e1_dot} admits the explicit solution
    \[
     \bar{e}_1(t)=\exp \left(-\lambda\int_0^t  \left(1 + \frac{v_1(\tau)}{\amax h}\right)\text{d}\tau\right)\bar{e}_1(0),
    \]
    showing that $\bar e_1(t)$ preserves its initial sign as long as $v_1>-\amax h$.
    This imply that $\bar e_1$ remains bounded as long as $v_1> -\amax h$. We now show that $v_1$ cannot reach $-\amax h$. By substituting the worst-case leader deceleration $u_0 = -\amax$, we obtain
    \begin{equation}
        \label{eq:v1_dot_bound}
        u_1
        \geq
        \frac{\lambda}{h}\bar e_1
        -
        \frac{\amax v_1}{\amax h+v_1}.
    \end{equation}
        Since $v_1$ is continuous, $v_{1,0}\in[0,\vmax]$, and $\bar e_1$ remains bounded while $v_1>-\amax h$, the second term in~\eqref{eq:v1_dot_bound} diverges to $+\infty$ as $v_1\to-\amax h$. Consequently, there exists $v_1^\star>-\amax h$ such that
        \[
            u_1>0,
            \qquad
            \forall v_1\le v_1^\star.
        \]
Therefore, whenever $v_1$ approaches $-\amax h$, one has $\dot v_1=u_1>0$, implying that the velocity is driven away from the boundary.
        
Assume now that the control saturates, so that $u_1=\sign(u_1)\amax$ and the error dynamics become
\[
      \dot{\bar{e}}_1 =v_0\left(1 +\frac{u_0}{\amax}\right) - \left( \amax h + v_1 \right) \sign(u_1)-v_1.
\]
If $\sign(u_1)=+1$, then $\dot{v}_1=\amax>0$, implying that $v_1$ increases and therefore remains strictly above $-\amax h$. If $\sign(u_1)=-1$ with $-\amax h <v_1<0$, then $\dot{\bar{e}}_1 \geq \amax h >0$,  meaning that $\bar{e}_1$ is increasing. From \eqref{eq:v1_dot_bound} we deduce that $u_1$ increases and there exists some $v^*>-\amax h$ such that   $\frac{\lambda}{h}\bar{e}_1 =\frac{\amax v^*}{\amax h + v^*} $
indicating that $v_1$ cannot escape in finite time, proving the first part of the statement.

We now prove positivity of $v_1$. Note that $v_1$ becomes negative if and only if at $v_1=0$ one has $\dot{v}_1=u_1<0$. Evaluating the control at $v_1=0$ when $u_1$ is not saturating, we obtain $u_1=\frac{\lambda}{h}\bar{e}_1+\frac{v_0(u_0+\amax)}{\amax h} \geq \frac{\lambda}{h}\bar{e}_1,$
hence non-negativity of $\bar{e}_1$ is a sufficient condition for the non-negativity of $v_1(t)$.  If the control does not saturate, we already showed that $\bar{e}_1$ does not change sign, thus $\bar{e}_1(0)\geq 0$ ensures that $v_1\geq 0$.\\ Consider now the saturated case. As established above, if $\sign(u_1)=-1$ then $\dot{\bar{e}}_1>0$ and the error cannot change from positive to negative. If $u_1=\amax$ then
\(
      \dot{\bar{e}}_1 =v_0\left(1 +u_0/\amax\right) - \left( 2v_1+\amax h \right). 
\)
For $\bar{e}_1=0$ the term in the saturation of~\eqref{eq:u1_full} is bigger than $\amax$, which rewrites as 
$
\frac{v_0(u_0+\amax)-v_1\amax}{\amax h+v_1}\geq \amax,
$
which, after direct algebraic manipulations yields $\left(v_0(1+u_0/\amax)-(2v_1+\amax h)\right) \geq 0$. Consequently, $\dot{\bar{e}}_1\geq 0$ and the error remains non-negative. Concluding, the system is well posed and if $\bar{e}_1(0)\geq 0$ then $v_1(t) \geq 0$ for all $t\geq 0$.
\end{proof}

Lemma~\ref{lem:well_posed} establishes two fundamental results. First, the well-posedness of the saturated control law \eqref{eq:u1_full} by strictly bounding the state away from the singularity $v_1 = -h\amax$. Second, the non-negativity of $v_1(t)$ assuming that $\bar{e}_1(0)\geq 0$. In the following subsection, we shall use these facts to prove that $\bar{e}_1(t)$ is globally asymptotically stable under \eqref{eq:u1_full}, and that it indeed recovers CACC configuration.

 \vspace{-0.4cm}
 \subsection{Main result}
 With this mathematical foundation secured, we now formally state the primary safety and tracking guarantees of the proposed controller.
\begin{theorem}\label{thm:safety}
Consider the leader-follower pair $(\Sigma_0, \Sigma_1)$ controlled by $u_1$ as defined in \eqref{eq:u1_full} and suppose  Assumption~\ref{asm:leader_bounds} holds. If the control gain is chosen such that:
\begin{equation}\label{eq:lambda_bound}
    \lambda \geq \frac{\amax h}{r-d_s},
\end{equation}
then the safe set $\mathcal{S}_1$ given in~\eqref{eq:safe-set} is forward invariant. Moreover, if the uncoordinated leader vehicle eventually maintains a constant cruising speed, such that $u_0(t)=0$ for $t \geq t_c$, the closed-loop system asymptotically recovers the standard CACC configuration, i.e., $\lim_{t\to+\infty} \bar{e}_1(t) = 0$ and $\lim_{t\to+\infty} \tilde{v}_1(t) = 0$.
\end{theorem}
\begin{proof}
By Nagumo's Theorem~\cite{blanchini1999set}, $\mathcal{S}_1$ is forward invariant if $\dot{s}_1 \geq 0$ whenever $s_1=0$. Let us show that $\dot{s}_1 \geq 0$ when $\tilde{v}_1 \geq 0$ and $\tilde{v}_1 < 0$ separately.

\emph{Case 1 ($\tilde{v}_1 \geq 0$):} Here, $q_1 \leq 0 \Rightarrow \max(0,q_1)=0$, reducing the safety function to $s_1 = \tilde{p}_1 - d_s$. On the boundary $s_1=0$, we have $\dot{s}_1 = \tilde{v}_1 \geq 0$, trivially satisfying the condition.

\emph{Case 2 ($\tilde{v}_1 < 0$):}  In this case, $q_1>0$ and $s_1 = \bar e_1 + r - d_s + hv_1.$ We evaluate $\dot s_1$ on the boundary $s_1=0$.
Assume first that the control is not saturated. Substituting~\eqref{eq:e1_dot} into the expression of $\dot s_1$ from~\eqref{eq:safety-function} and using that $u_0\ge-\amax$ yields
\begin{equation} \label{eq:dot-s1}
\dot s_1
\ge
\frac{\lambda v_1}{\amax h}(hv_1+r-d_s)
-
\frac{\amax hv_1}{\amax h+v_1}.
\end{equation}
Since $v_1>0$, enforcing $\dot s_1\ge0$ and dividing by $v_1$ yields the sufficient condition~\eqref{eq:lambda_bound}.
Assume now that the control is saturated. Let $u_{n}$ be the argument of the saturation operator in~\eqref{eq:u1_full}. We evaluate $u_n$ at the boundary to show that $u_1 = -\amax$. Using $\bar e_1=d_s-r-hv_1$ yields
\begin{equation}
\label{eq:u-nom}
u_n = -\lambda v_1
-\frac{\lambda}{h}(r-d_s)
+
\frac{v_0(\amax+u_0)-v_1\amax}{v_1+h\amax}.
\end{equation}
Because $\tilde v_1<0$, one has $v_1>v_0\ge0$, implying $v_0(\amax+u_0)-v_1\amax
<
v_0u_0
\le
v_0\amax$ and $\frac{v_0}{v_1+h\amax}<1$. Hence, using~\eqref{eq:u-nom}, $u_n < -\lambda v_1 - \lambda (r-d_s)/h+\amax$.
Using~\eqref{eq:lambda_bound}, we obtain that $u_n\leq-\lambda v_1$, which is strictly negative. 
Consequently, $u_1 = -\amax$.
Substituting this into the time derivative of the safety margin gives $\dot s_1=v_0\left(1+u_0/\amax\right),$ which is non-negative for all admissible $v_0$ and $u_0$ under Assumption~\ref{asm:leader_bounds}. Hence, $\mathcal S_1$ is forward invariant in all regimes. Now we show that we recover CACC configuration. Assume that there exists some $t_c\geq 0$ such that for $t \geq t_c$, the leader acceleration is null ($u_0(t) = 0$). Let us consider the following Lyapunov candidate:
\[
W = \int_0^{u_1} sds + \frac{\lambda}{2h}\tilde{v}_1^2,
\]
which is radially unbounded, positive semi-definite, and differentiable in $u_1$ and $v_1$. Taking its time derivative yields:
%
\[
\dot{W} = - \left( \lambda \left( 1 + \frac{v_1}{\amax h} \right) + \frac{\amax(v_0 + \amax h)}{(v_1 + \amax h)^2} \right) u_1^2.
\]
%
By Lemma \ref{lem:well_posed} the follower velocity is strictly lower-bounded by some $v^\star > -\amax h$, guaranteeing that the bracketed term is strictly positive. Therefore, $\dot{W} \leq 0$. By LaSalle's Invariance Principle, the system converges to the largest invariant set where $\dot{W} = 0$, implying $u_1 \to 0$.  Equating $u_1 = 0$ in \eqref{eq:u1_full} with $u_0 = 0$ implies that $\tilde{v}_1 \to 0$, which subsequently forces $\bar{e}_1 \to 0$, concluding the proof.
\end{proof}
These results can be extended to the entire platoon by applying Lemma \ref{lem:posVelSafe}, as shown in the following corollary.
\begin{corollary}
    Consider a platoon $(\Sigma_0,\ldots,\Sigma_N)$ satisfying assumptions \ref{asm:CACC-config}, \ref{asm:leader_bounds} and \ref{asm:cacc_stability}. If Theorem \ref{thm:safety} holds, then the entire platoon $(\Sigma_0,\ldots,\Sigma_N)$ recovers the standard CACC configuration. Moreover, if $\bar{e}_1(0)\geq 0$ then it is always safe.
\end{corollary}
\begin{proof}
 This follows directly from the inductive safety property of Lemma~\ref{lem:posVelSafe}, the positivity condition of Lemma~\ref{lem:well_posed}, and Theorem \ref{thm:safety}.
\end{proof}
Once the augmented error and relative velocity converge to the origin, recovering standard CACC configuration, a seamless transition to a classical CACC law can be performed to potentially recover long-term string stability. This dual-stage approach potentially allows $\Sigma_1$ to act as a specialized safety filter for the merging transient without compromising the attenuation of downstream disturbances during cruising. 
 \vspace{-0.15cm}
\section{Numerical example}\label{sec:ex}
To illustrate the potential of the proposed augmented control law, consider a platoon of $N=4$ vehicles in cruising mode. The platoon is initialized in a standard CACC configuration with $v_{i,0}=v_r=7 m/s$, where the followers ($i=2,3,4$) employ the classical CACC law:
%
\[
h\dot{u}_i= -u_i + Ke_i + Kh\dot{e}_i + u_{i-1}, \quad K=1/\sqrt{3},\; h=0.7
\]
%
as proposed in \cite[Sec. III]{PvdWN:14}. At $t=0$, a new vehicle $\Sigma_0$ joins the platoon at $v_{0,0}=0.5v_r$ and immediately brakes at full capacity until it reaches a standstill. Subsequently, its acceleration oscillates before eventually recovering to the cruising speed $v_r$. We compare the simulation results for two scenarios: i) the case where $u_1(t)$ is identical to the followers' LTI CACC, and ii) the case where $u_1(t)$ follows the proposed law \eqref{eq:u1_full}. \rev{Although we don't provide any theoretical guarantee,} in order to illustrate the potential of the proposed control under realistic scenarios, in both cases the feedforward element of the control is affected by a transmission delay, i.e. $u_{i-1}(t-\theta_i)$ with $\theta_i=0.1,0.1,0.2,0.15$. The constraints $v_i\in [0,\vmax=10]$ and $\abs{u_i}\leq \amax=4$ are treated as \emph{hard} constraints and are enforced for all vehicles. Consequently, velocity non-negativity is always satisfied. The initial condition for $\Sigma_0$ is chosen such that $s_1(0)=0$ exactly, the safety distance $d_s$ and standstill distance $r$ are chosen as $0.5[m]$ and $1[m]$ respectively. The gain $\lambda=5.6$ is chosen exactly at the bound.
\begin{figure}[!hbt]
 \centering
 \vspace{-0.4cm}
  \subfigure[Standard CACC for all vehicles.]{\label{fig:ex1:cacc}\includegraphics[width=0.49\columnwidth,clip]{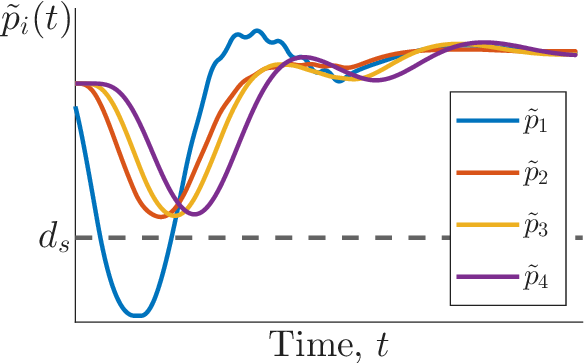}} \hfill
  \subfigure[First follower governed by \eqref{eq:u1_full}.]{\label{fig:ex1:nlc}\includegraphics[width=0.49\columnwidth,clip]{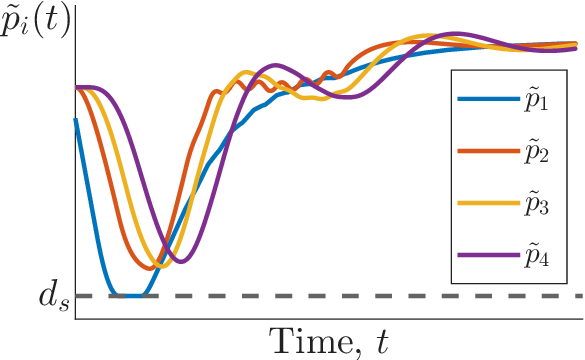}}
 \caption{Relative positions $\tilde{p}_i(t)$ for the platoon.} \label{fig:ex1}
 \vspace{-0.25cm}
\end{figure}
As shown in Fig.~\ref{fig:ex1:cacc}, the standard LTI controller fails to maintain safety despite the initial distance being $\tilde{p}_1(0)\approx 5 r$. In contrast, Fig.~\ref{fig:ex1:nlc} demonstrates that our approach ensures that the entire platoon remains safe throughout the aggressive transient. Fig.~\ref{fig:ex1:u} compares the control signals for both setups where, as expected, both scenarios display string-stable behavior for $i>1$. However, as seen in Fig.~\ref{fig:ex1:nlc_u}, the control law \eqref{eq:u1_full} displays a better tracking behavior of $u_0$, resulting in smaller oscillations in the following vehicles. Notably, safety is \rev{maintained in this simulation} despite $\bar{e}_1(0)=-5.4$ and the heterogeneous transmission delays\rev{, illustrating the potential robustness of the control scheme}.
\begin{figure}[!hbt] 
\vspace{-0.4cm}
 \centering
  \subfigure[Standard CACC for all vehicles.]{\label{fig:ex1:cacc_u}\includegraphics[width=0.49\columnwidth,clip]{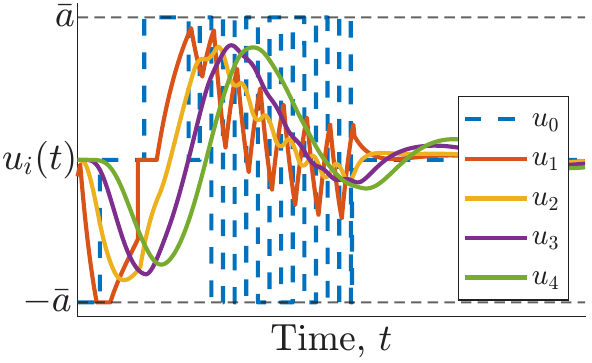}} \hfill
  \subfigure[First follower governed by \eqref{eq:u1_full}.]{\label{fig:ex1:nlc_u}\includegraphics[width=0.49\columnwidth,clip]{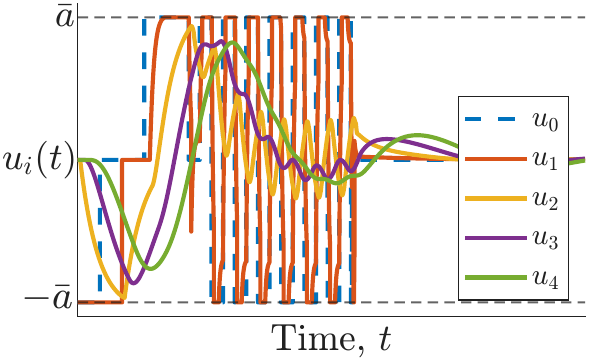}}
 \caption{Control signal (acceleration) of the leader (dashed) and the subsequent platoon.} \label{fig:ex1:u}
\end{figure}
\vspace{-0.6cm}
\section{Concluding remarks}\label{sec:conc}
%
We proposed a novel safe-by-design platoon merging architecture that augments the standard Constant Time Headway (CTH) policy with a velocity-dependent quadratic term. Derived from a worst-case analysis under bounded controls, this term ensures safe trajectories always exist. We designed a control law that maintains a simple Cooperative Adaptive Cruise Control (CACC) structure, renders the safe set forward invariant, and asymptotically recovers the standard CTH spacing as velocities synchronize. Furthermore, we demonstrated that under nominal conditions, safety is an inductive property: modifying only the first follower's control law guarantees downstream platoon safety, provided velocities remain non-negative. Because preliminary analysis indicates this non-negativity condition remains necessary under non-ideal conditions, future work will focus on enforcing this state constraint during control synthesis and conducting a formal robustness analysis to relax the nominal operation assumption (\ref{asm:cacc_stability}). 

\vspace{-0.20cm}
\bibliographystyle{IEEEtran}
\bibliography{MyBib}

\newcommand{\noopsort}[1]{}
\begin{thebibliography}{10}
\providecommand{\url}[1]{#1}
\csname url@samestyle\endcsname
\providecommand{\newblock}{\relax}
\providecommand{\bibinfo}[2]{#2}
\providecommand{\BIBentrySTDinterwordspacing}{\spaceskip=0pt\relax}
\providecommand{\BIBentryALTinterwordstretchfactor}{4}
\providecommand{\BIBentryALTinterwordspacing}{\spaceskip=\fontdimen2\font plus
\BIBentryALTinterwordstretchfactor\fontdimen3\font minus
  \fontdimen4\font\relax}
\providecommand{\BIBforeignlanguage}[2]{{%
\expandafter\ifx\csname l@#1\endcsname\relax
\typeout{** WARNING: IEEEtran.bst: No hyphenation pattern has been}%
\typeout{** loaded for the language `#1'. Using the pattern for}%
\typeout{** the default language instead.}%
\else
\language=\csname l@#1\endcsname
\fi
#2}}
\providecommand{\BIBdecl}{\relax}
\BIBdecl

\bibitem{SBS:79}
S.~Sklar, J.~Bevans, and G.~Stein, ``"safe-approach" vehicle-follower
  control,'' \emph{IEEE Trans.\ Vehicular Tech.}, vol.~28, no.~1, pp. 56--62,
  1979.

\bibitem{CI:92}
C.~C. Chien and P.~Ioannou, ``Automatic vehicle-following,'' in \emph{Proc. of
  American Control Conf.}, 1992, pp. 1748--1752.

\bibitem{SH:96}
D.~Swaroop and J.~Hedrick, ``String stability of interconnected systems,''
  \emph{IEEE Trans.\ Automat.\ Control}, vol.~41, no.~3, pp. 349--357, 1996.

\bibitem{KM:09}
S.~Klinge and R.~H. Middleton, ``String stability analysis of homogeneous
  linear unidirectionally connected systems with nonzero initial conditions,''
  in \emph{IET Irish Signals and Systems Conference}, 2009, pp. 1--6.

\bibitem{IC:93}
P.~Ioannou and C.~Chien, ``Autonomous intelligent cruise control,'' \emph{IEEE
  Trans.\ Vehicular Tech.}, vol.~42, no.~4, pp. 657--672, 1993.

\bibitem{DPH:17}
V.~S. Dolk, J.~Ploeg, and W.~P. M.~H. Heemels, ``Event-triggered control for
  string-stable vehicle platooning,'' \emph{IEEE Trans.\ Intelligent Transport.
  Syst.}, vol.~18, no.~12, pp. 3486--3500, 2017.

\bibitem{SB-L:24}
A.~Samii and N.~Bekiaris-Liberis, ``Predictor-based cacc design for
  heterogeneous vehicles with distinct input delays,'' \emph{IEEE Open Journal
  of Intelligent Transport. Sys.}, vol.~5, pp. 783--796, 2024.

\bibitem{MPD:25}
G.~Ma, P.~Pagilla, and S.~Darbha, ``Parasitic actuation delay limits the
  minimum employable time headway in connected and autonomous vehicles,''
  \emph{ArXiv}, 2025.

\bibitem{PSNWN:11}
J.~Ploeg, B.~T.~M. Scheepers, E.~van Nunen, N.~van~de Wouw, and H.~Nijmeijer,
  ``Design and experimental evaluation of cooperative adaptive cruise
  control,'' in \emph{Proc. of 14th IEEE Conf.\ Intelligent and Transport.
  Syst.}, 2011, pp. 260--265.

\bibitem{MSSNKN:14}
V.~Milanés, S.~E. Shladover, J.~Spring, C.~Nowakowski, H.~Kawazoe, and
  M.~Nakamura, ``Cooperative adaptive cruise control in real traffic
  situations,'' \emph{IEEE Trans.\ Intelligent Transport. Syst.}, vol.~15,
  no.~1, pp. 296--305, 2014.

\bibitem{MPD:24}
G.~Ma, P.~Pagilla, and S.~Darbha, ``Assessing the safety benefits of cacc+
  based coordination of connected and autonomous vehicle platoons in emergency
  braking scenarios,'' in \emph{IEEE Intelligent Vehicles Symposium}, 2024, pp.
  2248--2254.

\bibitem{GH:24}
R.~Gaagai and J.~Horn, ``Distributed safety-critical control for linear
  homogeneous vehicle platoons subject to actuator and communication delays,''
  in \emph{63rd IEEE Conf.\ Decision and Control}, 2024, pp. 8376--8383.

\bibitem{CTJM:24}
X.~Chen, Z.~Tang, K.~H. Johansson, and J.~Mårtensson, ``Safe platooning and
  merging control using constructive barrier feedback,'' \emph{European Journal
  of Control}, vol.~80, p. 101060, 2024.

\bibitem{MA:24}
S.~Mair and M.~Althoff, ``Provably correct safety protocol for cooperative
  platooning,'' in \emph{IEEE Intelligent Veh. Symposium}, 2024, pp. 780--787.

\bibitem{LS-KPAN:18}
J.~Ligthart, E.~Semsar-Kazerooni, J.~Ploeg, M.~Alirezaei, and H.~Nijmeijer,
  ``Controller design for cooperative driving with guaranteed safe behavior,''
  in \emph{2018 IEEE Conference on Control Technology and Applications (CCTA)},
  2018, pp. 1460--1465.

\bibitem{mehra2015adaptive}
A.~Mehra, X.~Ma, F.~Berg, P.~Tabuada, J.~Grizzle, and A.~Ames, ``Adaptive
  cruise control: {E}xperimental validation of advanced controllers on
  scale-model cars,'' in \emph{ACC}, 2015, pp. 1411--1418.

\bibitem{SLSM:25}
R.~Su, A.~Lahmadi, Y.-Q. Song, and J.-P. Mangeot, ``Assessing 5g connectivity
  for urbanloop: a pod-based autonomous railway transport system,'' in
  \emph{102nd IEEE Vehicular Tech. Conf.}\hskip 1em plus 0.5em minus
  0.4em\relax IEEE, 2025, pp. 1--7.

\bibitem{wvu_prt_2025}
{West Virginia University}, ``{WVU} personal rapid transit system
  https://prt.wvu.edu,'' 1975.

\bibitem{wang2026velocity}
W.~Wang, J.~Kreiss, P.~Lorenzetti, L.~Licitra, G.~Lefebvre, and R.~Postoyan,
  ``{Velocity Tracking for Autonomous Railway-based Urbanloop Pods by
  Contraction},'' in \emph{Proc.\ 23rd IFAC World Congress}, 2026.

\bibitem{PvdWN:14}
J.~Ploeg, N.~van~de Wouw, and H.~Nijmeijer, ``Lp string stability of cascaded
  systems: Application to vehicle platooning,'' \emph{IEEE Trans.\ Control
  Syst.\ Technol.}, vol.~22, no.~2, pp. 786--793, 2014.

\bibitem{blanchini1999set}
F.~Blanchini, ``Set invariance in control,'' \emph{Automatica}, vol.~35,
  no.~11, pp. 1747--1767, 1999.

\end{thebibliography}


\appendices
\gdef\thesection{\Alph{section}}
\makeatletter
\renewcommand\@seccntformat[1]{\appendixname\csname the#1\endcsname.\hspace{0.5em}}
\makeatother
\end{document}